\documentclass{article}

\usepackage{amsmath} 
\usepackage{amssymb}  
\usepackage{mathtools}		
\usepackage{amsthm}
\usepackage{xcolor}
\usepackage{balance}
\usepackage{subcaption}
\usepackage[margin=1.5in]{geometry}
\usepackage{hyperref}
\usepackage[ruled]{algorithm2e}

\newcommand{\R}{\mathbb{R}} 
\newcommand{\K}{\mathbb{K}} 
 
\newcommand{\N}{\mathbb{N}}

\newcommand{\Lie}{\mathcal{L}}

\newcommand{\M}{\mathbb{M}}

\DeclarePairedDelimiter{\abs}{\lvert}{\rvert}
\DeclarePairedDelimiter{\norm}{\lVert}{\rVert}

\DeclarePairedDelimiterX{\inp}[2]{\langle}{\rangle}{#1, #2}
\DeclarePairedDelimiter{\supp}{\textrm{supp}(}{)}
\DeclarePairedDelimiter{\Mp}{\mathcal{M}_+(}{)}

\newtheorem{thm}{Theorem}[section]

\newtheorem{rmk}{Remark} 

\usepackage{color}
\usepackage{ulem}

\title{\LARGE \bf
Peak Estimation and Recovery \\ with Occupation Measures
}

\author{Jared Miller$^1$, Didier Henrion $^2$, Mario Sznaier$^1$
\thanks{$^1$J. Miller and M. Sznaier are with the   ECE Department, Northeastern University, Boston, MA 02115. (e-mails: miller.jare@northeastern.edu,  msznaier@coe.neu.edu). They were partially supported by NSF grants  CNS--1646121, CMMI--1638234, ECCS--1808381 and CNS--2038493,  and AFOSR grant FA9550-19-1-0005.}
\thanks{$^2$D. Henrion is with LAAS-CNRS, Universit\'e de Toulouse, CNRS,  France and the Faculty of Electrical Engineering of the Czech Technical University in Prague, Czechia. (e-mail: henrion@laas.fr)}
}

\begin{document}

\maketitle
\pagestyle{plain}

\begin{abstract}
\label{sec:abstract}
Peak Estimation aims to find the maximum value of a state function achieved by a dynamical system. This problem is non-convex when considering standard Barrier and Density methods for invariant sets, and has been treated heuristically by using auxiliary functions. A convex formulation based on occupation measures is proposed in this paper to solve peak estimation. This method  is dual to the auxiliary function approach. Our method will converge to the optimal solution and can recover trajectories even from approximate solutions. This framework is extended to safety analysis by maximizing the minimum of a set of costs along trajectories.
\end{abstract}

\section{Introduction}
\label{sec:intro}

 The behavior of dynamical systems may be analyzed by bounding extreme values of state functions along trajectories. 
For a system with dynamics governed by an ODE $\dot{x} = f(t, x)$ with continuous $f$, let $x(t \mid x_0)$ denote a trajectory starting from an initial point $x_0$. 
The problem of finding the maximum value of a function $p(x)$ for trajectories starting from a set $X_0$ evolving over the time interval $[0, T]$ is \begin{equation}
    \label{eq:peak_traj}
    \begin{aligned}
    P^* = & \max_{t,\: x_0 \in X_0} p(x(t)) \\
    & \dot{x}(t) = f(t, x), \quad t \in [0, T]. \\
    \end{aligned}
\end{equation}

The goal of peak estimation is to approximate sharp upper bounds to $P^*$. It is also desired to recover the near-optimal trajectories that achieve $p(x(t \mid x_0)) \approx P^*$ for some time $t \in [0, T]$. 
Lower bounds to $P^*$ can be found by sampling an initial point $x_0 \in X_0$ and finding the maximum value of $p(x)$ along $x(t \mid x_0)$, but generating a sampled lower bound that is close to $P^*$ is difficult. 
Upper bounds of $P^*$ are universal properties of all trajectories, and $P^*$ may be sandwiched between discovered lower and upper bounds. Peak estimation may be infinite-time if $T = \infty$.

Problem \eqref{eq:peak_traj} was cast into an infinite-dimensional linear program (LP) of occupation measures in the context of optimal stopping problems of a martingale in \cite{cho2002linear}, and the bound $P^*$ was approximated by discretization with finite-dimensional LPs. The infinite dimensional LP in \cite{cho2002linear} is an extension to the stochastic setting of the deterministic optimal control formulation in \cite{lewis1980relaxation} with a state cost instead of a running cost. 
A survey of infinite-dimensional LP methods is available at \cite{fattorini1999infinite}, and LPs in occupation measures may also be solved through the moment-SOS hierarchy of Semidefinite Programs (SDP) \cite{lasserre2008nonlinear}.
More recently, an auxiliary function approach was developed 
to find a convergent sequence of upper bounds to $P^*$ by sum-of-squares (SOS) methods \cite{fantuzzi2020bounding}. The infinite-dimensional LP in \cite{fantuzzi2020bounding} that is truncated into an SOS program is dual to the infinite-dimensional LP in \cite{cho2002linear}. The optimal trajectories that achieve $P^*$ are localized into a sublevel set of the solved auxiliary function, and may be approximated through adjoint optimization \cite{gunzburger2002perspectives, fantuzzi2020bounding}. The SOS programs in \cite{fantuzzi2020bounding} are dual to Linear Matrix Inequalities (LMIs) in moments of occupation measures \cite{lasserre2008nonlinear, henrion2013convex}.

This paper has two main focuses: recovery and safety. A recovery algorithm is detailed to find the optimal trajectories from a peak estimation problem, which is based on solutions of occupation measure LMIs that satisfy an approximate rank constraint.  Peak estimation is extended to maximin objectives, which aim to maximize the minimum of a set of cost functions. Maximin peak estimation may be used to quantify a safety margin of trajectories with respect to an  unsafe set, complementing the binary safety determination of Barrier \cite{prajna2004safety, prajna2006barrier} and Density \cite{rantzer2004analysis} functions.

This paper is organized as follows: Section \ref{sec:prelim} defines notation and reviews preliminaries. Section \ref{sec:recovery} posits a recovery algorithm to extract near-optimal trajectories.
Section \ref{sec:merged_safety} provides a safety evaluation framework via maximin optimization.
Section \ref{sec:conclusion} concludes the paper. \section{Preliminaries}
\label{sec:prelim}
\subsection{Notation}
\label{sec:prelim_notation}

$\R$ is the space of real numbers, $\R^n$ is an $n$-dimensional real vector space, and $\R^n_+$ is its nonnegative real orthant. 
Let $x = (x_1 \ldots x_n)$ be a tuple of $n$ independent variables. Monomials may be expressed as $x^\alpha = \prod_{i=1}^n x_i^{\alpha_i}$ for a multi-index $\alpha \in \N^n$, where $\N$ is the set of natural numbers.
A polynomial $p(x)$ may be expressed as  $\sum_{\alpha \in \mathcal{I}} p_\alpha x^\alpha$ for a finite index set $\mathcal{I}$, and the degree of $p$ is the maximum $\abs{\alpha} = \sum_i \abs{\alpha_i}$ over all $\alpha \in \mathcal{I}$. $\R[x]$ is the ring of polynomials in $x$, and $\R[x]_{\leq d}$ is the subset of polynomials with total degree at most $d$.
A basic semialgebraic set is the locus of $N_c$ inequality constraints $g_k(x) \geq 0$ of bounded degree, with the form $\K = \{x \mid g_k(x) \geq 0 , \ k = 1\ldots N_c\}$.


Assume for this paper that $X \subset \R^m$ for some dimension $m$.
Let $C(X)$ be the set of continuous functions defined over $X$, and $C^1(X)$ be the subset of $C(X)$  with continuous first derivatives.
The space of finite signed Borel measures over $X$ is $\mathcal{M}(X)$, which is the topological dual of $C(X)$
with duality pairing $\inp{f}{\mu}$ = $\int_X f (x) d\mu(x) = \int_X f d\mu$ for $f\in C(X)$ and $\mu \in \mathcal{M}(X)$ if $X$ is compact.
The nonnegative subcones $C_+(X)$ and $\Mp{X}$  are dual cones with an induced inner product $\inp{\cdot}{\cdot}$ from the duality pairing.
If $B \subseteq X$ and $I_B(x)$ is the indicator function on $B$, then the measure of $B$ with respect to $\mu$ is $\mu(B) = \int_X I_B(x) d \mu = \int_B d \mu$. 
The support $\supp{\mu}$ is the smallest closed subset $S \subseteq X$ such that $\mu(X \setminus  S) = 0$. 
$\mu_1 \otimes \mu_2$ denotes the product measure formed by $\mu_1$ and $\mu_2$. $\mu$ is a probability measure on $X$ if $\mu(X) = 1$. The Dirac delta is a probability measure $\delta_x \in \Mp{X}$  with $\supp{\delta_x} = x$.

\subsection{Moment-SOS Hierarchy}
\label{sec:moment_sos}


Infinite-dimensional LPs may be defined over nonnegative measures $\mu$, such as 
\begin{equation}
\label{eq:meas_program} p^* =  \max_{\mu \in \Mp{X}} \inp{c}{\mu}, \qquad  \mathcal{A}(\mu) = b.
\end{equation}
In program \eqref{eq:meas_program}, $c \in \R[x] \subset C(X)$ is a cost, and $\mathcal{A}, b$ define a set of affine constraints in the moments of $\mu$. 
The $\alpha$-moment of a measure $\mu$ is the scalar quantity $y_\alpha = \inp{x^\alpha}{\mu}$.
A measure $\mu$ from \eqref{eq:meas_program} may be parameterized by an infinite sequence of moments $y = \{y_\alpha \mid \ \forall \alpha \in \N^n\} $ such that $\inp{p}{\mu} = \sum_{\alpha} p_\alpha y_\alpha$ for all  $p \in \R[x]$.  
Problem \eqref{eq:meas_program} may be expressed as a linear program in the infinite number of moments $y_\alpha$, which
must be truncated into a problem with a finite number of variables and constraints for tractable optimization. The moment-sum-of-squares (SOS) hierarchy uses  the moment sequence $y_\alpha$ with degree $\abs{\alpha} \leq d$ for some bound $d$ as variables. Refer to \cite{lasserre2009moments} for all material in this subsection. 

Let $X = \{x \mid g_k(x) \geq 0\}$ be a basic semialgebraic set with $N_c$ constraints where each $g_k(x)$ has degree $d_k$, and $\alpha, \beta, \gamma \in \N^m$ be a set of multi-indices. 
The infinite moment matrix $\M(y)$ and localizing matrices $\M(g_{k} y)$ for each inequality constraint are indexed by multi-indices $(\alpha, \beta)$, and have the forms,  
\begin{equation}
    \M(y)_{\alpha \beta} = y_{\alpha + \beta}, \quad \M(g_{k} y)_{\alpha \beta}  =  \textstyle\sum_{\gamma} g_{k \gamma} y_{\alpha + \beta + \gamma}.
\end{equation}

The degree-$d$ relaxation of problem \eqref{eq:meas_program} is:
\begin{subequations}
\begin{align}
\label{eq:mom_program}
    p^*_d &=  \ \max_{y} \textstyle\sum_\alpha c_\alpha y_\alpha, \qquad  A(y) = b_y \\
    & \M_d(y) \succeq 0, \; \M_{d - d_k}(g_{k} y) \succeq 0  \; \forall k = 1, \ldots, N_c.
\end{align}
\end{subequations}

The truncated moment matrix $\M_d(y)$ and localizing matrices $\M_{d-d_k}(g_k y)$ contain moments of up to order $2d$. The expression $A(y)=b_y$ in problem \eqref{eq:mom_program} imposes the affine condition $\mathcal{A}(\mu) = b$ on entries of the moment sequence $y$.
The optima of relaxations $p^*_{d} \geq p^*_{d+1} \geq p^*_{d+2} \ldots$ are a sequence of upper bounds to the true optimum $p^*$. 
This sequence will converge as $d \rightarrow \infty$ if $X$ is compact, or more specifically an Archimedean condition holds \cite{lasserre2009moments}. The size of $\M_d(y)$ (size of the PSD matrix variable of an LMI absent other structure) with $n$ variables at degree $d$ is $\binom{n+d}{d}$.

A measure is `rank-$r$ atomic' if it is supported on a set of $r$ discrete points called `atoms'. A moment matrix $\M_d(y)$ possesses a representing measure $\mu$ if the moments of $\mu$ up to order $2d$ agree with the entries of $\M_d(y)$. 
A necessary but not sufficient condition for a rank-$r$ atomic measure $\mu$ to exist is that $\M_d(y)$ has rank $r$. The sufficient condition requires the existence of a flat extension, in which the rank of truncated moment matrices is preserved as $d$ increases \cite{laurent2009generalized}.
The $r$ atoms may be recovered by a Cholesky decomposition of $\M_d(y)$, or by reading entries of $y_\alpha$ if $r=1$ \cite{henrion2005detecting}.


\subsection{Occupation Measures}
\label{sec:prelim_occ}

Occupation measures are a valuable tool in solving optimal control and reachable set problems. Resources on this topic include \cite{lasserre2008nonlinear, korda2014convex}. This section follows the exposition of \cite{josz2019transient}. For a single initial point $x_0 \in X_0$, the occupation measure $\mu(A \times B \mid x_0)$ is the amount of time the trajectory $x(t \mid x_0)$ spends in the region $A \times B \subseteq [0, T] \times X$:
\begin{equation}
\label{eq:occ_measure_single}
\mu(A \times B \mid x_0) = \int_0^T I_{A \times B}(t, x(t \mid x_0)) dt.
\end{equation}
\indent The average occupation measure $\mu$ yields the $\mu_0$-weighted time trajectories spend in $A \times B$ for some $\mu_0 \in \Mp{X_0}$:
\begin{equation}
\label{eq:occ_measure}
    \mu(A \times B) = \int_{X_0} \mu(A \times B \mid x_0) d \mu_0(x_0).
\end{equation}
\indent It holds that $\mu([0, T] \times X) = T$. The final occupation measure is the distribution of $x \in X$ that results after following initial conditions distributed as $\mu_0$ for time $T$
\begin{equation}
\label{eq:final_measure}
    \mu_T(B) = \int_{X_0} I_B(x (T \mid x_0)) d \mu_0(x_0).
\end{equation}
\indent 
For a test function $v(t, x) \in C^1([0, T]\times X)$, 
the Lie derivative operator $\Lie_f$ is defined
\begin{equation}
    \label{eq:lie_deriavite}
    \Lie_f v(t,x) = \partial_t v(t, x)+ \nabla_x v(t,x)^T f(t, x).
\end{equation}
\indent The three measures $\mu_0, \mu_T, \mu$ are linked together by the linear Liouville Equation $\delta_T \otimes\mu_T =  \delta_0 \otimes\mu_0 + \Lie_f^\dagger \mu$, which may be understood in a weak sense to hold
for all test functions $v(t,x) \in C^1([0, T] \times X)$,
\begin{equation}
\label{eq:liou_int}
    \inp{v(T, x)}{\mu_T} = \inp{v(0, x)}{\mu_0} + \inp{\Lie_f v(t,x)}{\mu}.
\end{equation}
The operator $\Lie_f^\dagger$ is the adjoint of $\Lie_f$ such that $\inp{\Lie_f v}{\mu}=\inp{v}{\Lie_f^\dagger \mu}$ for any $v(t,x) \in C^1([0, T] \times X)$. 
\subsection{Peak Estimation}
\label{sec:peak_occ}

Peak estimation problems can be bounded by an infinite-dimensional LP in measures by defining a peak measure $\mu_p \in \Mp{[0, T] \times X}$, which generalizes $\delta_T \otimes \mu_T$ with free terminal time.
Eq. (9) from \cite{cho2002linear} with variables $(\mu_0, \mu, \mu_p)$ can be restated as
\begin{subequations}
\label{eq:peak_meas}
\begin{align}
p^* = & \ \textrm{max} \quad \inp{p(x)}{\mu_p} \label{eq:peak_meas_obj} \\
    & \mu_p = \delta_0 \otimes\mu_0 + \Lie_f^\dagger \mu \label{eq:peak_meas_flow}\\
    & \mu_0(X_0) = 1 \label{eq:peak_meas_prob}\\
    & \mu, \mu_p \in \Mp{[0, T] \times X} \label{eq:peak_meas_peak}\\
    & \mu_0 \in \Mp{X_0}. \label{eq:peak_meas_init}
\end{align}
\end{subequations}
The probability measure $\mu_0$ in \eqref{eq:peak_meas_prob} is distributed over initial conditions. By Liouville's Equation \eqref{eq:peak_meas_flow}, $\mu_p$ is a probability measure over points in time and space:
\begin{equation}
\label{eq:liou_rel_one}
\inp{1}{\mu_p} = \inp{1}{\delta_0 \otimes \mu_0} + \inp{ \Lie_f (1)}{\mu} = 1 + 0 = 1.
\end{equation}
The dual problem to \eqref{eq:peak_meas} with variables $(v(t,x), \gamma)$ is
\begin{subequations}
\label{eq:peak_cont}
\begin{align}
    d^* = & \ \min_{\gamma \in \R} \quad \gamma & \\
    & \gamma \geq v(0, x)  & &  \forall x \in X_0 \label{eq:peak_cont_init}\\
    & \Lie_f v(t, x) \leq 0 & & \forall (t, x) \in [0, T] \times X \label{eq:peak_cont_f}\\
    & v(t, x) \geq p(x) & & \forall (t, x) \in [0, T] \times X \label{eq:peak_cont_p} \\
    &v \in C^1([0, T]\times X) &
\end{align}
\end{subequations}
and is formulated in Eq. 2.5 and 2.6 of \cite{fantuzzi2020bounding}. The auxiliary function $v(t,x)$ and scalar $\gamma$ are dual variables for constraints \eqref{eq:peak_meas_flow} and \eqref{eq:peak_meas_prob} respectively \cite{fantuzzi2020bounding}.
If $(v, \gamma)$ solves to \eqref{eq:peak_cont}, then the sublevel set $\{(t, x) \mid v(t, x) \leq \gamma\}$ contains all trajectories starting from $X_0$.

The measures in \eqref{eq:peak_meas} have bounded mass when $[0, T] \times X$ is compact (requires finite $T$), and the image of the cone of measures $\Mp{[0, T] \times X)^2} \times \Mp{X_0}$ through the linear mapping in \eqref{eq:peak_meas_flow} \eqref{eq:peak_meas_prob} is closed in the weak star topology. This implies the strong duality $p^* = d^*$ as shown in Theorem C.20 of \cite{lasserre2009moments}.
The solution $p^* = d^* \geq P^*$ is an upper bound for the true peak in \eqref{eq:peak_traj}. 
The solution $p^*$ is approximately equal to $P^*$ for compact $[0, T] \times X$ and locally Lipschitz dynamics (Thm 2.1 of \cite{lewis1980relaxation}, 2.5 of \cite{fantuzzi2020bounding}), and often $p^*=P^*$.
The objective values $d^*=P^*$ are  tight if the function $v(t,x)$ in \eqref{eq:peak_cont} is allowed to be discontinuous \cite{fantuzzi2020bounding}.


\cite{cho2002linear} estimates \eqref{eq:peak_meas} by discretizing the infinite-dimensional LP (sec. 4.1) or forming a Markov chain (sec. 4.2). \cite{fantuzzi2020bounding} finds a convergent sequence of upper bounds through  an SOS relaxation (Eq. 4.4-4.7).

For numerical examples in this paper, it is assumed that $p$ and the entries of $f$ are given polynomials, and
\vspace{-0.03in} 
\begin{subequations}
\begin{align}
    X = &\{x \mid \ g_k(x) \geq 0, \ \forall k = 1, \ldots, N_c\}\\
    X_0 = &\{x \mid \ g_{0k}(x) \geq 0, \forall k = 1, \ldots, N_c^0 \}
\end{align}
\end{subequations}
are compact basic semialgebraic sets.

\section{Recovery}
\label{sec:recovery}

    This section presents an algorithm to attempt extraction of optimal trajectories if $p^*$ is reached at $R$ points.
\subsection{Optimal Trajectories and Measures}
Each  of the $R$ solution trajectories to Problem \eqref{eq:peak_traj} that achieves $P^*$  may be encoded by a triple $(x_0^r, t_p^r, x_p^r)$ satisfying $P^* = p(x_p^r) = p(x(t_p^r \mid x_0^r))$ for $r = 1, \ldots, R$. A trajectory $x(t \mid x_0)$ in which $P^*$ is reached multiple times is separated into triples for each attainment. 

Let the triple $(x_0, t_p, x_p)$ be a solution to Problem \eqref{eq:peak_traj}. The probability measures $\mu_0 = \delta_{x_0}$, $\mu_p = \delta_{t_p}\otimes \delta_{ x_p}$, and $\mu$ defined by  Eq. \eqref{eq:occ_measure_single} with an endpoint $t_p$ instead of $T$ satisfy constraints \eqref{eq:peak_meas_flow}-\eqref{eq:peak_meas_init} with an objective value of $\inp{p}{\mu_p} = P^*$ (where $\mu$ is supported between $ (0, x_0)$ and $(t_p,x_p)$).
For the general case where $P^*$ is reached at multiple triples $(x_0^r, t_p^r, x_p^r)$, the measures $\mu_0 = \sum_{r=1}^r w_r \delta_{x^r_0}$,  $\mu_p = \sum_{r=1}^R w_r (\delta_{t_p^r}\otimes \delta_{ x_p^r})$ , and $\mu = \sum_{r=1}^R w_r \mu^r$ are feasible solutions to \eqref{eq:peak_meas_flow}-\eqref{eq:peak_meas_init} for all weights $w \in \R^{R}_+$ with $\mathbf{1}^T w = 1$ (convex combinations). Optimal trajectories may be recovered from the support of $\mu_0$ and $\mu_p$ solving \eqref{eq:peak_meas} if $p^*=P^*$.

\subsection{LMI Formulation}
Assume that the measures $\mu_0, \mu, \mu_p$ from \eqref{eq:peak_meas} have moment sequences of $y^0, y, y^p$ up to degree $2d$.  Liouville's equation in \eqref{eq:peak_meas_flow} implies that the following linear relation holds for each test function $v(t,x) = x^\alpha t^\beta$,
\begin{equation}
\label{eq:liou_mom}
    \inp{x^\alpha t^\beta}{ \delta_0 \otimes \mu_0} + \inp{\Lie_f( x^\alpha t^\beta)}{\mu} - \inp{x^\alpha t^\beta}{\mu_p} = 0.
\end{equation}
Define $\textrm{Liou}_{\alpha \beta}(y^0, y, y^p)$ as the relation in moment sequences from \eqref{eq:liou_mom} for each test function $x^\alpha t^\beta$.
The degree-$d$ LMI relaxation of \eqref{eq:peak_meas} with variables $(y^0, y, y^p)$ is 
\begin{subequations}
\label{eq:peak_lmi}
\begin{align}
    p^*_d = & \textrm{max} \quad \textstyle\sum_{\alpha} p_\alpha y_\alpha^p. \label{eq:peak_lmi_obj} \\
    & \textrm{Liou}_{\alpha \beta}(y^0, y, y^p) = 0 \quad \forall (\alpha, \beta) \in \N^{m+1}_{\leq 2d} \label{eq:peak_lmi_flow}\\
    & y^0_0 = 1 \\
    & \M_d(y^0), \M_d(y), \M_d(y^p) \succeq 0 \\
    &\M_{d - d_{0k}}(g_{0k} y^0) \succeq 0  \label{eq:peak_lmi_init} \ \; \forall k = 1, \ldots, N_c^0\\ 
    &\M_{d - d_k}(g_{k} y) \succeq 0  \, \  \quad \; \forall k = 1, \ldots, N_c \label{eq:peak_lmi_peak} \\ 
    &\M_{d - d_k}(g_{k} y^p) \succeq 0  \ \quad \forall k = 1, \ldots, N_c \\ 
    &\M_{d-2}(t(T-t)y),   \ \M_{d-2}(t(T-t)y^p) {\; \succeq \;} 0. \label{eq:peak_lmi_supp2}
\end{align}
\end{subequations}

\noindent Program \eqref{eq:peak_lmi} is dual to the degree-$d$ SOS program in \cite{fantuzzi2020bounding}. Localizing matrix constraints \eqref{eq:peak_lmi_init}-\eqref{eq:peak_lmi_supp2} enforce the measure support constraints in \eqref{eq:peak_meas_peak}-\eqref{eq:peak_meas_init}.

\subsection{Recovery Algorithm}

A solution to \eqref{eq:peak_lmi} at degree $d$ will yield an upper bound $p^*_d \geq p^*$. If the moment matrices $\M_d (y^0)$  and $\M_d (y^p)$ are rank-deficient and have a flat extension, then there  exist atomic representing measures $\mu_0, \ \mu_p$ whose measures agree with the moment sequences up to order $2d$. These representing measures may not necessarily solve \eqref{eq:peak_meas}, as there may not exist a $\mu$ supported on the graph of optimal trajectories with moments in $\M_d(y)$. 

The atoms of $\M_d (y^0)$  and $\M_d (y^p)$ with extraction by \cite{henrion2005detecting} (or reading $y^0, \ y^p$ if rank-1, which automatically implies existence of a flat extension) are candidates for optimal triples $(x_0^r, t_p^r, x_p^r)$. Evaluating $p(x)$ along a sampled trajectory starting at a feasible atom $x_0^r \in X_0$ from $\M_d (y^0)$ will yield a lower bound $p_d^r$ such that $p_d^r \leq p^* \leq p_d^*$.  If the lower and upper bound are sufficiently close together, then the trajectory starting at $x_0^r$ is approximately-optimal. Algorithm \ref{alg:recovery} describes the forward trajectory recovery algorithm. An alternative approach could take atoms from $\mu_p$ with $p^*_d  - p(x_p^{r}) \leq \epsilon$,  running $f$ backwards from $x_p^{r}$ for time $t_p^{r}$ and observing if the destination point is a member of $X_0$.  

This process assumes that the peak estimation problem takes an optimal value at a finite set of points, which in practice is not very restrictive. The rank-recovery process requires low rank moment matrices, is sensitive to numerical conditioning in the monomial basis as $d$ increases, and may not always succeed (e.g. $p^* - P^* > \epsilon$).

\SetKwFor{Loop}{Loop}{}{EndLoop}
\begin{algorithm}[ht]

         \textbf{Input :} Sets $X_0, \; \ X$, dynamics $f$, cost $p$, max. time $T$, initial degree $d_0$, tolerance $\epsilon$\\ 
         \textbf{Output :} Near-Optimal Trajectories $OPT$\\
         degree $d = d_0$, optimal triples $OPT = \varnothing$\\
         \Loop{}{
            Solve \eqref{eq:peak_lmi} at degree $d$ for $(p^*_d, \  \M_d(y^0))$ \\
            \uIf{$\M_d(y^0)$ has a flat extension}{
                \For{atoms $x_0^r$ in $\M_d(y^0)$ by \cite{henrion2005detecting}}{
                Simulate $x(t \mid x_0^r)$  \\
                Find $p^r_d = \max_{t\in[0, T]} p(x(t \mid x_0^r))$ \\
               Find $t_p^r, \ x_p^r$ on traj. with $p^r_d = p(x_p^r)$\\
                \uIf{$p^*_d - p^r_d < \epsilon$}{ Append $(x_0^r, t_p^r, x_p^r)$ to $OPT$}
                }
            }
            \Return $OPT$ \textbf{if} $OPT \neq \varnothing$\\
            $d \leftarrow d + 1$
         }
         \caption{Trajectory recovery}
         \label{alg:recovery}
        \end{algorithm}    


Example 4.1 from \cite{fantuzzi2020bounding} is the following system with a central symmetry and two stable attractors:
\begin{equation}
\label{eq:sym_attract_sys}
    \begin{bmatrix}
    \dot{x}_1 \\ \dot{x}_2
    \end{bmatrix} = 
    \begin{bmatrix}
    0.2 x_1 + x_2 - x_2(x_1^2+x_2^2) \\
    -0.4 x_1 + x_1(x_1^2+x_2^2)
    \end{bmatrix}
\end{equation}

For the  infinite-horizon problem (without variable $t$) of maximizing  $\norm{x}_2^2$ starting at $X_0  = \{x \mid \ \norm{x}_2^2 = 0.5\}, \; X = [-2, 2]^2$, \eqref{eq:peak_lmi} finds a bound $p^*_7 =1.90318$. The solved $\M_7(y_0), \M_7(y_p)$ are rank-2 up to a tolerance of $3\times10^{-4}$. When using Alg. \ref{alg:recovery}, $p^*_7$ is within $0.005$ of the sampled result $p_7^r$ of each atom.

Fig. \ref{fig:sym_attract_sys_traj} plots the optimal trajectory in dark blue and randomly sampled trajectories in cyan along with the level set $p(x) = p^*_7$ in the red dashed line.
The black dashed curve is the level set $\{x \mid v(x) = 0\}$.
Fig. \ref{fig:sym_attract_sys_init} compares the extracted $x_0^*\approx\pm(0.491,-0.093)$ (blue circles) and $x_p^*\approx\pm(0.481, 1.293)$ (blue stars) against a sublevel-set approximation to locations of optimal trajectories and their initial conditions (\cite{fantuzzi2020bounding} Sec. 3: $\{x \mid 0 \leq v(x) + p_7^* \leq 0.002, \ 0 \leq \Lie_f v(x) \leq 0.004\}$). 
\begin{figure}[!ht]
     \centering
     \begin{subfigure}[b]{0.53\linewidth}
         \centering
        \includegraphics[width = \linewidth]{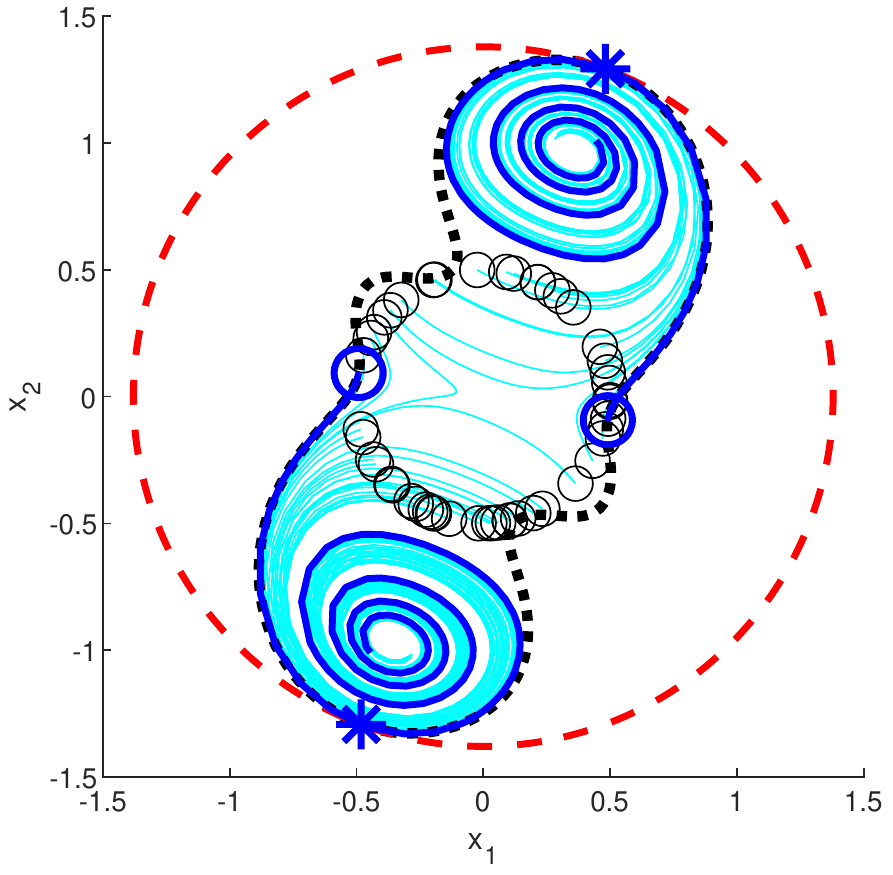}
         \caption{Trajectories and Bounds}
         \label{fig:sym_attract_sys_traj}
     \end{subfigure}
     \;
     \begin{subfigure}[b]{0.35\linewidth}
         \centering
         \includegraphics[width=\linewidth]{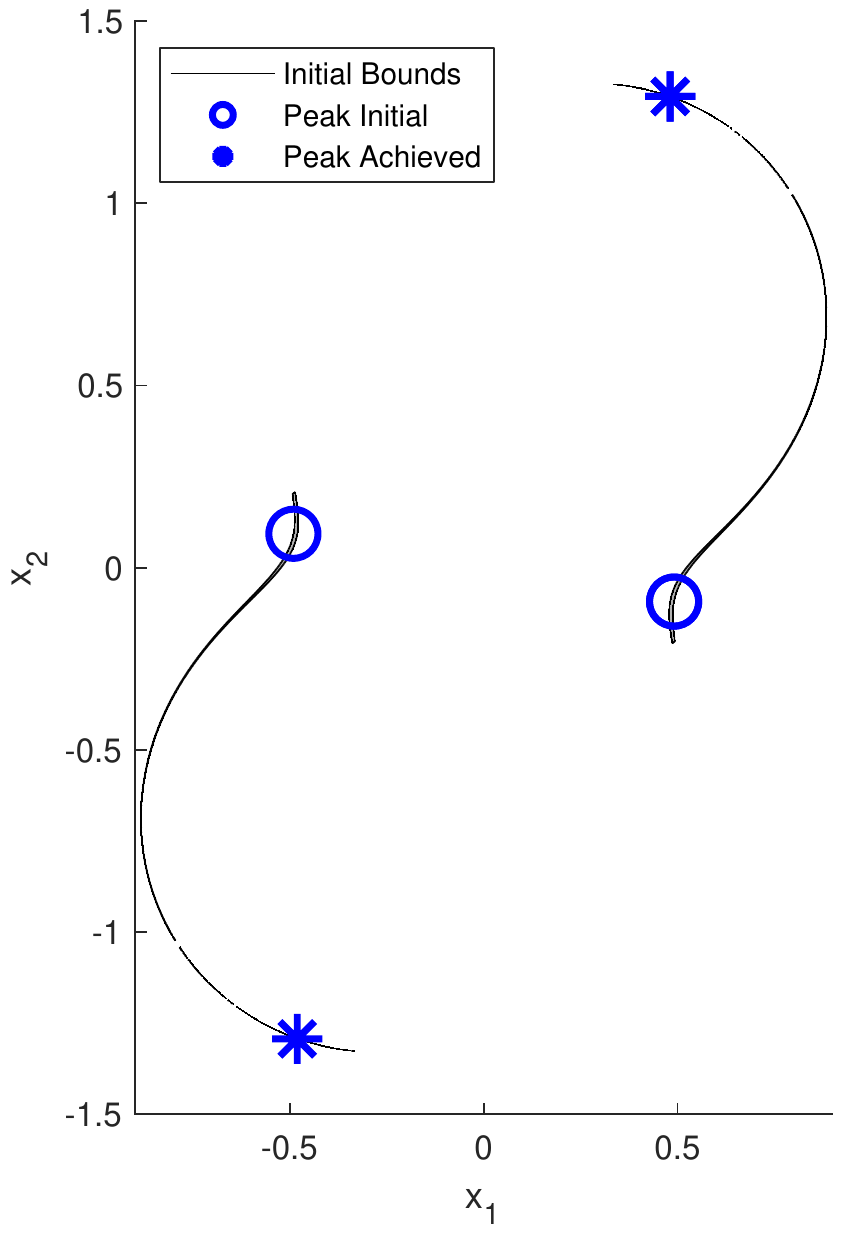}
         \caption{near-optimal set}
         \label{fig:sym_attract_sys_init}
     \end{subfigure}
         \caption{\label{fig:sym_attract_sys} Maximize $\norm{x}_2^2$ along \eqref{eq:sym_attract_sys}}
\end{figure}

Code in this paper is publicly available at
\url{https://github.com/jarmill/peak/} and was written in MATLAB 2020a. SDPs were formulated in YALMIP \cite{lofberg2004yalmip} and Gloptipoly 3 \cite{henrion2003gloptipoly},  and solved with Mosek 9.2.

\section{Maximin Safety Analysis}
\label{sec:merged_safety}
This section proposes a method to analyze safety of trajectories through maximin peak estimation.
\label{sec:safety}

\subsection{Safety Background}

Let $X_u\subset X$ be an unsafe compact basic semialgebraic set with description $X_u = \{x \mid p_i(x) \geq 0 \ \forall i = 1 \ldots N_u\}$. For example,  $X_u$ could represent a the location of a body of water in a driving task. Barrier functions for autonomous dynamics $\dot{x} = f(x)$ offer one method of certifying safety \cite{prajna2004safety} by attempting to find a $B(x) \in C^1(X)$ satisfying,
\begin{subequations}
\label{eq:barrier}
\begin{align}
    &B(x) > 0 & &  \forall x \in X_u\\
    &B(x) \leq 0 & & \forall  x \in X_0 \\
    &f(x) \cdot \nabla_x B(x) \leq 0 & & \forall  x \in X
\end{align}
\end{subequations}
A measure problem with variables $(\mu_0, \mu, \mu_u)$ is,
\begin{subequations}
\label{eq:unsafe}
\begin{align}
    &\mu_u = \mu_0 + \Lie_f^\dagger \mu \qquad \qquad  \mu_0(X_0) = 1 \\
    \mu_0 \in &\Mp{X_0}, \ \mu_u \in \Mp{X_u}, \ \mu \in \Mp{X}.
\end{align}
\end{subequations}

Problems \eqref{eq:barrier} and \eqref{eq:unsafe} are strong alternatives.
Feasibility of \eqref{eq:unsafe} implies that there exists at least one trajectory starting in $X_0$ and ending in the unsafe set $X_u$. If there exists a barrier function $B$ satisfying \eqref{eq:barrier}, then no trajectory starting from $X_0$ enters $X_u$. Finding a valid $B$ certifies safety but does not indicate proximity of trajectories to $X_u$. 

All points $x \in X_u$ satisfy $p_i(x) \geq 0$, and therefore $X_u = \{x \mid \min_i p_i(x) \geq 0\}$. If $\min_i p_i(x) < 0$ for all points on trajectories coming from $X_0$, then trajectories never enter or contact $X_u$. A negative maximum value of $\min_i p_i(x)$ on trajectories therefore certifies safety, which can be verified through peak estimation.

\subsection{Maximin Objective}
\label{sec:maximin}



Let $p(x) = [p_i(x)]_{i=1}^{N_p}$ be a polynomial vector of objectives. The maximin peak estimation problem is 
\begin{equation}
    \label{eq:peak_traj_multi}
    \begin{aligned}
    P^* = & \max_{t,\:x_0 \in X_0}  \ \min_i p_i(x) \\
    & \dot{x}(t) = f(t, x), \quad t \in [0, T]. \\
    \end{aligned}
\end{equation}

\begin{thm}
The maximin peak estimation problem \eqref{eq:peak_traj_multi} may be upper bounded by a measure program
\begin{subequations}
\label{eq:peak_meas_multi}
\begin{align}
    p^* = &\ \max \quad q &  \\
    & q +z_i =  \inp{p_i(x)}{\mu_p} \qquad \quad \forall i=1\ldots N_p \label{eq:peak_meas_multi_ineq}\\
    & \mu_p = \delta_0 \otimes\mu_0 + \Lie_f^\dagger \mu  \label{eq:peak_meas_multi_flow}& \\
    & \mu_0(X_0) = 1 \label{eq:peak_meas_multi_prob}& \\
    & q \in \R, \ z \in \R_+^{N_p} & \\
    & \mu, \mu_p \in \Mp{[0, T] \times X}&\\
    & \mu_0 \in \Mp{X_0}. & \label{eq:peak_meas_multi_init}
\end{align}
\end{subequations}
\end{thm}

\begin{proof}This is an extension to the measure program \eqref{eq:peak_meas} upper bounding \eqref{eq:peak_traj} with multiple costs. The value $q$ is a lower bound on $\inp{p_i}{\mu_p}$, and Program \eqref{eq:peak_meas_multi} aims to find the maximum such $q$. Nonnegative slack variables $z_i$ in \eqref{eq:peak_meas_multi_ineq} fill the gap between the bound $q$ and $\inp{p_i}{\mu_p}$.\end{proof}
Degree-$d$ LMI relaxations provide a decreasing sequence of upper bounds to $p^*$ in \eqref{eq:peak_meas_multi}.

The Lagrangian of \eqref{eq:peak_meas_multi} is
\begin{align}
\label{eq:multi_lagrangian}
    L = & \gamma(\inp{1}{\mu_p} - 1) + \textstyle\sum_{i=1}^{N_p}{ \alpha_i z_i + \beta_i (q + z_i - \inp{p_i}{\mu_p}}) \nonumber \\
    &+ \inp{v(t,x)}{\delta_0 \otimes \mu_0+\Lie_f^\dagger \mu - \mu_p} + q
\end{align}
with new dual variables $\beta \in \R^{N_p}$ from constraint \eqref{eq:peak_meas_multi_ineq} and $\alpha \in \R_+^{N_p}$ from the cone constraint $z \in \R_+^{N_p}$. After eliminating $\alpha$, the dual problem to \eqref{eq:peak_meas_multi} is:
\begin{subequations}
\label{eq:peak_cont_multi}
\begin{align}
    d^* = & \ \min_{\gamma \in \R} \quad \gamma & \\
    & \gamma \geq v(0, x) &  &  \forall x \in X_0 \label{eq:peak_cont_init_multi}\\
    & \Lie_f v(t, x) \leq 0 & &  \forall (t, x) \in [0, T] \times X \label{eq:peak_cont_f_multi}\\
    & v(t, x) \geq \beta^T p(x) & & \forall (t, x) \in [0, T] \times X \label{eq:peak_cont_p_multi} \\
    &v \in C^1([0, T]\times X)& \\
    &  \beta \in \R^{N_p}_{+}, \ \mathbf{1}^T \beta = 1&\label{eq:peak_cont_beta_multi}
\end{align}
\end{subequations}
Strong duality holds between \eqref{eq:peak_meas_multi} and \eqref{eq:peak_cont_multi} by Theorem C.20 of \cite{lasserre2009moments} when $[0, T] \times X$ is compact. 
\begin{rmk}
If a particular term $p_i(x)$ is minimal among $p(x)$ at optimality, then $z_i = 0$ and $\beta_i \neq 0$. The dual variable $\beta$ is located on an $N_p$-dimensional simplex, so a single-objective case will feature $\beta = 1$.
\end{rmk}




\subsection{Maximin Example}
An example of maximin estimation is the following non-autonomous ODE (Example 2.1 from \cite{fantuzzi2020bounding}):

\begin{equation}
\label{eq:time_var_sys}
    \begin{bmatrix}
    \dot{x}_1 \\ \dot{x}_2
    \end{bmatrix} = 
    \begin{bmatrix}
    x_2 t - 0.1 x_1 - x_1 x_2 \\ x_1 t - x_2 + x_1^2
    \end{bmatrix}
\end{equation}

Figure \ref{fig:time_var_sys} plots trajectories from equation \eqref{eq:time_var_sys} on the initial set $X_0=\{x \mid \  (x_1+0.75)^2+x_2^2=1\}$ and total set $X = [-3, 2]\times [-2, 2]$. When maximizing $p(x) = x_1$, over the time range $[0, 5]$ the first three bounds are:
\[p^{*}_{1:3} = [1.5473, \; 0.4981, \; 0.4931]\]
The second-largest eigenvalue of $\M_1(y) = 2.943\times 10^{-6}$, so the moment matrix is nearly rank-1 for atom extraction by Algorithm \ref{alg:recovery}. The near-optimal trajectory is displayed in Fig. \ref{fig:time_var_sys_single}
with $x_0^* = [-1.674, -0.383]$ and $x_p^* = [0.493, 0.029]$. 
With a maximin objective $p(x) = [x_1, \; x_2]$, the first three bounds are
\[p^{*}_{1:3} = [1.0765, \; 0.3905, \; 0.3891]\]
At $d=3$, the optimal $\beta = [0.647, 0.353]$ has both elements nonzero, as $p_1(x_p^*) = p_2(x_p^*) = p^{*}_3.$
Fig. \ref{fig:time_var_sys_multi_cost} displays the maximin objective $\min(x_1, x_2)$ along trajectories in Fig. \ref{fig:time_var_sys}.
$x_p^*$ is reached at time $t_p^* = 2.19$, which is indicated by the blue stars on Fig. \ref{fig:time_var_sys_multi_cost}.


\begin{figure}[ht]
     \centering
     \begin{subfigure}[b]{0.48\linewidth}
         \centering
         \includegraphics[width=\linewidth]{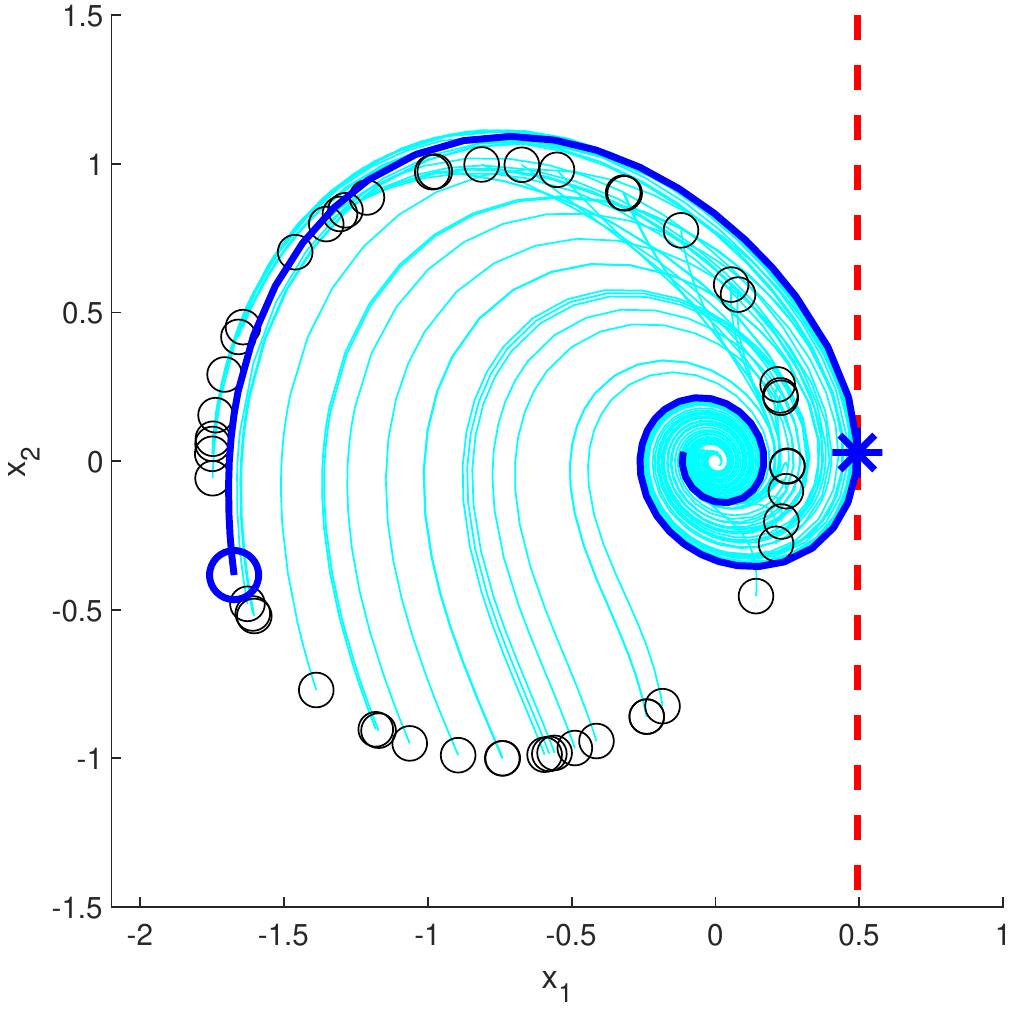}
         \caption{$\max x_1$}
         \label{fig:time_var_sys_single}
     \end{subfigure}
     \;
     \begin{subfigure}[b]{0.48\linewidth}
         \centering
         \includegraphics[width=\linewidth]{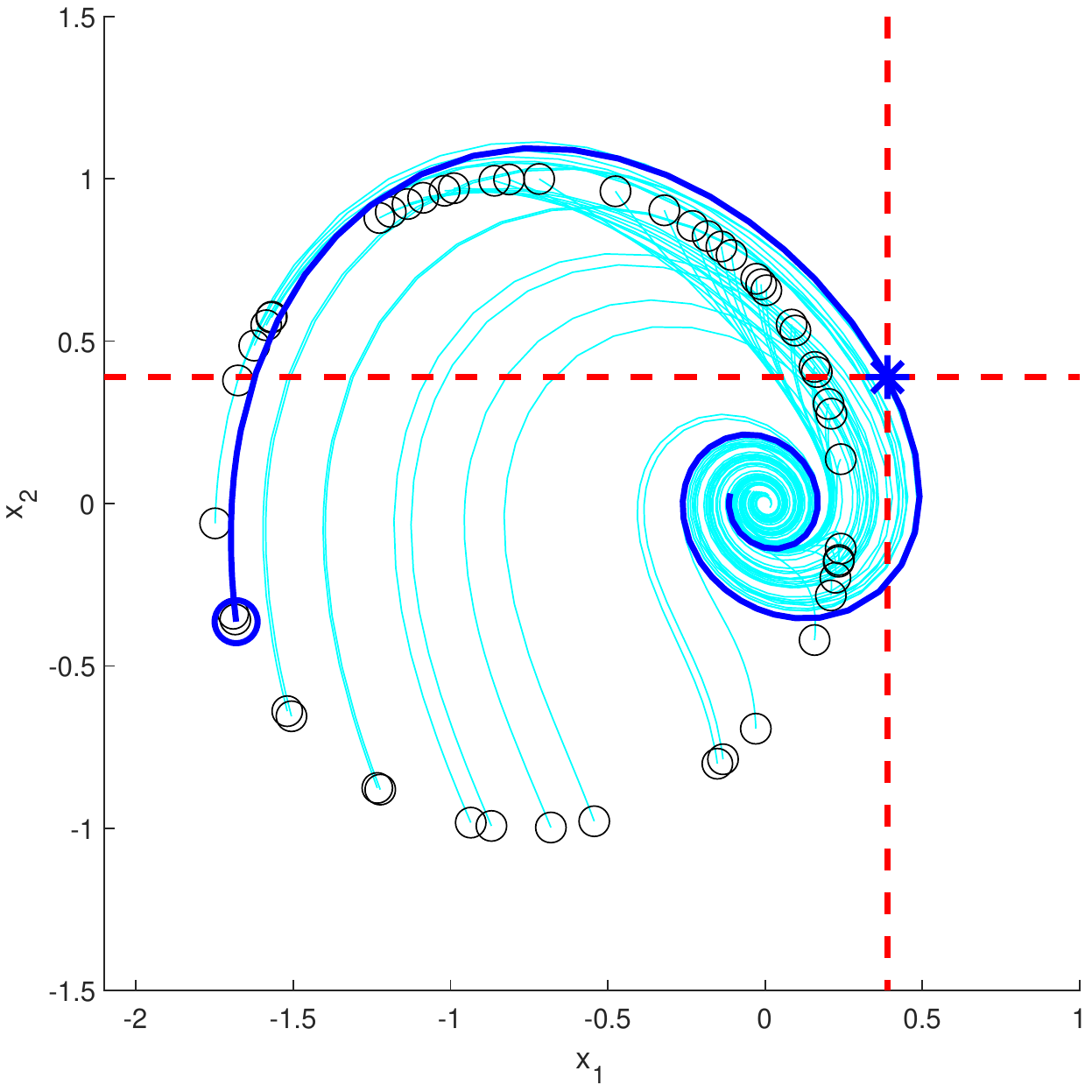}
         \caption{$\max \min(x_1, x_2)$}
         \label{fig:time_var_sys_multi}
     \end{subfigure}
         \caption{\label{fig:time_var_sys} Peak analysis of system \eqref{eq:time_var_sys} at $d=3$}
\end{figure}

\begin{figure}[ht]
    \centering
    \includegraphics[width=\linewidth]{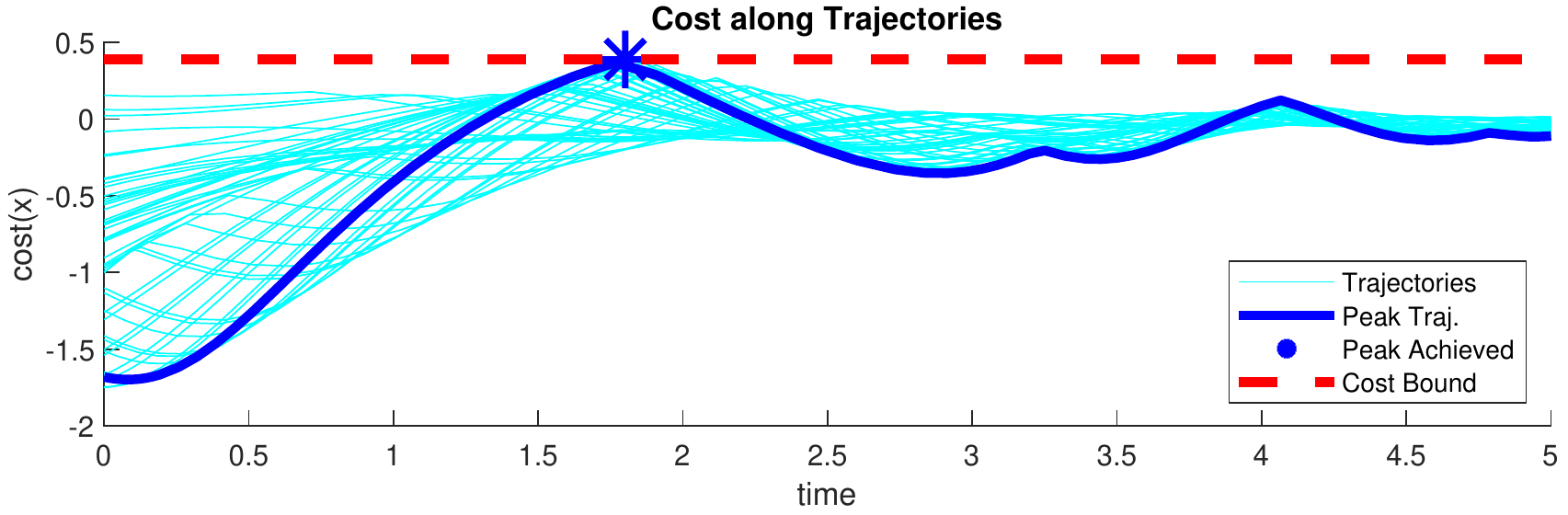}
    \caption{    \label{fig:time_var_sys_multi_cost} The value of  $\min(x_1, x_2)$ along trajectories \eqref{eq:time_var_sys}}
\end{figure}


\subsection{Safety Margins}

Let $P^*$ and $p_d^*$ be the optimum value of the maximin problem \eqref{eq:peak_traj_multi} with costs $[p_i(x)]_{i=1}^{N_u}$ from $X_u$ and \eqref{eq:peak_traj_multi}'s degree-$d$ LMI relaxation \eqref{eq:peak_meas_multi} respectively.
A bound $0 > p_d^* \geq P^*$ for some $d$ certifies safety, and this negative value $p_d^*$ serves as a `safety margin'. In contrast, a unsafe $P^* > 0$ implies that $p_d^* > 0$ for all degrees $d$.




Figure \ref{fig:flow_margin} demonstrates safety margins on the system $f(x) = [x_2, \  -x_1 -x_2 + x_1^3/3]$ from \cite{prajna2004safety} with an infinite-horizon. 
Trajectories originate from $X_0: (x_1 - 1.5)^2 + x_2^2 \leq 0.4^2$. The unsafe set is a red half-circle formed by a circle with radius $R_u = 0.5$ centered at $C_u = [0, -0.5]$ cut by a half-space $[\cos{\theta}, \sin{\theta}]'[x-C_u] \geq 0$ for some angle $\theta$. With $\theta =  5\pi/4$ in Fig. \ref{fig:flow_safe}, $p_{3:5}^* = [0.1178, -0.1326, -0.1417]$. All trajectories are safe because  $p_4^* < 0$.

The flow in Fig. \ref{fig:flow_unsafe} with $\theta = 3\pi/4$ is unsafe, as some trajectory (the approximate optimal trajectory with $d=5$ recovered by Alg. \ref{alg:recovery}) passes through $X_u$. Feasibility of the LMI associated with \eqref{eq:unsafe} with $\mu \in \Mp{[0, T] \times X}, \ \mu_u \in \Mp{[0, T] \times X_u}$ is generally a more reliable way to certify unsafety than safety margins with Alg. \ref{alg:recovery}, as the recovery routine requires (near) flat extensions.

\begin{figure}[t]
     \centering
     \begin{subfigure}[b]{0.48\linewidth}
         \centering
         \includegraphics[width=\linewidth]{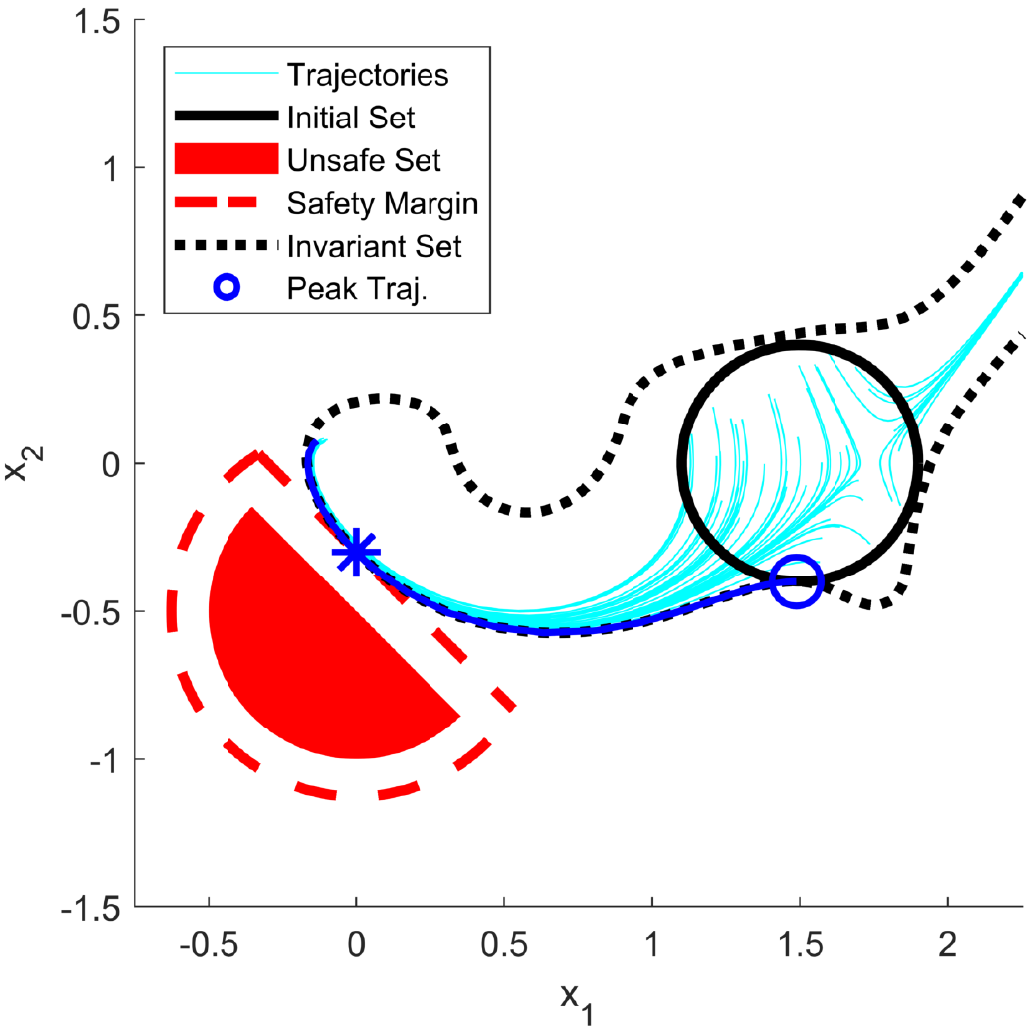}
         \caption{Safe: $p^*_5 = -0.1417 < 0$}
         \label{fig:flow_safe}
     \end{subfigure}
     \;
     \begin{subfigure}[b]{0.48\linewidth}
         \centering
         \includegraphics[width=\linewidth]{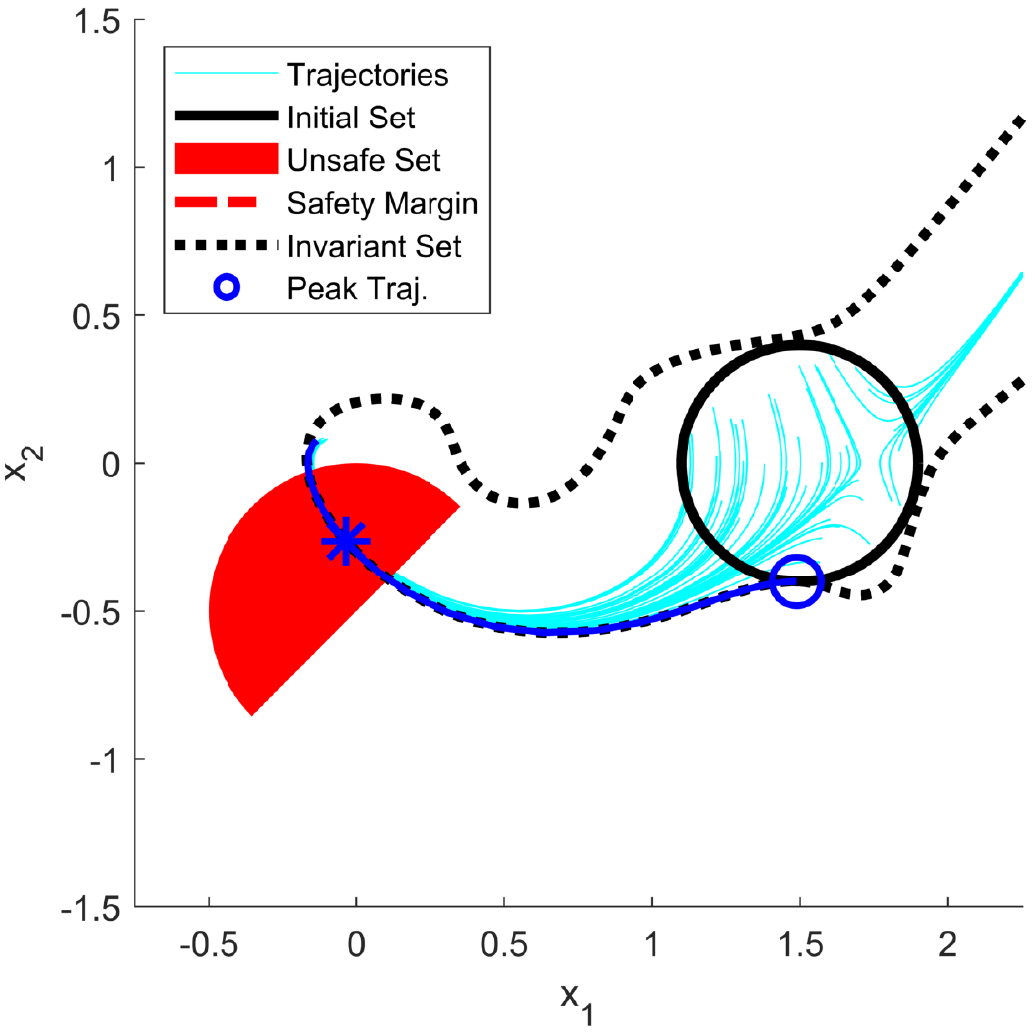}
         \caption{Unsafe: $p^*_5 = 0.1935 > 0$}
         \label{fig:flow_unsafe}
     \end{subfigure}
     \caption{\label{fig:flow_margin} Safety margins for half-circle sets}
\end{figure}




\section{Conclusion}
\label{sec:conclusion}

This paper presented a recovery algorithm to extract optimal trajectories of a peak estimation problem based on rank-deficiency of moment matrices. The peak estimation framework is extended to maximin objectives. A safety analysis framework is presented with maximin peak estimation, where a negative value of a safety margin is sufficient to certify safety of trajectories.
Future work includes bounding minimum distances to unsafe sets and peak estimation under uncertainty.

\section*{Acknowledgements}

The authors thank Milan Korda for his helpful discussions about the details of Algorithm \ref{alg:recovery}.

\bibliographystyle{IEEEtran}
\bibliography{peak_letter_reference.bib}

\begin{thebibliography}{10}
\providecommand{\url}[1]{#1}
\csname url@samestyle\endcsname
\providecommand{\newblock}{\relax}
\providecommand{\bibinfo}[2]{#2}
\providecommand{\BIBentrySTDinterwordspacing}{\spaceskip=0pt\relax}
\providecommand{\BIBentryALTinterwordstretchfactor}{4}
\providecommand{\BIBentryALTinterwordspacing}{\spaceskip=\fontdimen2\font plus
\BIBentryALTinterwordstretchfactor\fontdimen3\font minus
  \fontdimen4\font\relax}
\providecommand{\BIBforeignlanguage}[2]{{%
\expandafter\ifx\csname l@#1\endcsname\relax
\typeout{** WARNING: IEEEtran.bst: No hyphenation pattern has been}%
\typeout{** loaded for the language `#1'. Using the pattern for}%
\typeout{** the default language instead.}%
\else
\language=\csname l@#1\endcsname
\fi
#2}}
\providecommand{\BIBdecl}{\relax}
\BIBdecl

\bibitem{cho2002linear}
M.~J. Cho and R.~H. Stockbridge, ``Linear programming formulation for optimal
  stopping problems,'' \emph{SIAM Journal on Control and Optimization},
  vol.~40, no.~6, pp. 1965--1982, 2002.

\bibitem{lewis1980relaxation}
R.~Lewis and R.~Vinter, ``Relaxation of optimal control problems to equivalent
  convex programs,'' \emph{Journal of Mathematical Analysis and Applications},
  vol.~74, no.~2, pp. 475--493, 1980.

\bibitem{fattorini1999infinite}
H.~O. Fattorini, H.~O. Fattorini \emph{et~al.}, \emph{Infinite dimensional
  optimization and control theory}.\hskip 1em plus 0.5em minus 0.4em\relax
  Cambridge University Press, 1999, vol.~54.

\bibitem{lasserre2008nonlinear}
J.~B. Lasserre, D.~Henrion, C.~Prieur, and E.~Tr{\'e}lat, ``{Nonlinear Optimal
  Control via Occupation Measures and LMI-Relaxations},'' \emph{SIAM Journal on
  Control and Optimization}, vol.~47, no.~4, pp. 1643--1666, 2008.

\bibitem{fantuzzi2020bounding}
G.~Fantuzzi and D.~Goluskin, ``{Bounding Extreme Events in Nonlinear Dynamics
  Using Convex Optimization},'' \emph{SIAM Journal on Applied Dynamical
  Systems}, vol.~19, no.~3, pp. 1823--1864, 2020.

\bibitem{gunzburger2002perspectives}
M.~D. Gunzburger, \emph{Perspectives in flow control and optimization}.\hskip
  1em plus 0.5em minus 0.4em\relax SIAM, 2002.

\bibitem{henrion2013convex}
D.~Henrion and M.~Korda, ``{Convex Computation of the Region of Attraction of
  Polynomial Control Systems},'' \emph{IEEE TAC}, vol.~59, no.~2, pp. 297--312,
  2013.

\bibitem{prajna2004safety}
S.~Prajna and A.~Jadbabaie, ``{Safety Verification of Hybrid Systems Using
  Barrier Certificates},'' in \emph{International Workshop on Hybrid Systems:
  Computation and Control}.\hskip 1em plus 0.5em minus 0.4em\relax Springer,
  2004, pp. 477--492.

\bibitem{prajna2006barrier}
S.~Prajna, ``Barrier certificates for nonlinear model validation,''
  \emph{Automatica}, vol.~42, no.~1, pp. 117--126, 2006.

\bibitem{rantzer2004analysis}
A.~Rantzer and S.~Prajna, ``{On Analysis and Synthesis of Safe Control Laws},''
  in \emph{42nd Allerton Conference on Communication, Control, and
  Computing}.\hskip 1em plus 0.5em minus 0.4em\relax University of Illinois,
  2004, pp. 1468--1476.

\bibitem{lasserre2009moments}
J.~B. Lasserre, \emph{{Moments, Positive Polynomials And Their}
  {Applications}}, ser. Imperial College Press Optimization Series.\hskip 1em
  plus 0.5em minus 0.4em\relax World Scientific Publishing Company, 2009.

\bibitem{laurent2009generalized}
M.~Laurent and B.~Mourrain, ``A generalized flat extension theorem for moment
  matrices,'' \emph{Archiv der Mathematik}, vol.~93, no.~1, pp. 87--98, 2009.

\bibitem{henrion2005detecting}
D.~Henrion and J.-B. Lasserre, ``Detecting global optimality and extracting
  solutions in gloptipoly,'' in \emph{Positive polynomials in control}.\hskip
  1em plus 0.5em minus 0.4em\relax Springer, 2005, pp. 293--310.

\bibitem{korda2014convex}
M.~Korda, D.~Henrion, and C.~N. Jones, ``{Convex Computation of the Maximum
  Controlled Invariant Set For Polynomial Control Systems},'' \emph{SIAM
  Journal on Control and Optimization}, vol.~52, no.~5, pp. 2944--2969, 2014.

\bibitem{josz2019transient}
C.~Josz, D.~K. Molzahn, M.~Tacchi, and S.~Sojoudi, ``{Transient Stability
  Analysis of Power Systems via Occupation Measures},'' in \emph{2019 IEEE
  Power \& Energy Society Innovative Smart Grid Technologies Conference
  (ISGT)}.\hskip 1em plus 0.5em minus 0.4em\relax IEEE, 2019, pp. 1--5.

\bibitem{lofberg2004yalmip}
J.~{Lofberg}, ``Yalmip : a toolbox for modeling and optimization in matlab,''
  in \emph{ICRA (IEEE Cat. No.04CH37508)}, 2004, pp. 284--289.

\bibitem{henrion2003gloptipoly}
D.~Henrion and J.-B. Lasserre, ``{GloptiPoly: Global Optimization over
  Polynomials with Matlab and SeDuMi},'' \emph{ACM Transactions on Mathematical
  Software (TOMS)}, vol.~29, no.~2, pp. 165--194, 2003.

\end{thebibliography}
\end{document}